\documentclass[reqno]{amsart}

\usepackage{amssymb,epsfig,graphics,fullpage}
\usepackage{amsmath}
\usepackage{amsthm}
\usepackage{amscd}
\usepackage{verbatim}
\input xypic

\def\r{\rangle}
\def\l{\langle}
\def\R{\mathbb R}
\def\Z{\mathbb Z}
\def\a{\alpha}
\def\w{\omega}
\newtheorem{theorem}{Theorem}
\newtheorem{lemma}[theorem]{Lemma}

\def\Imm{\operatorname{Imm}}
\def\stab{\operatorname{stab}}
\def\per{\operatorname{per}}

\begin{document}

\title[]
 {On immanant functions related to Weyl groups of $A_n$ }

\author{Lenka H\'akov\'a$^1$, Agnieszka Tereszkiewicz$^2$}


\begin{abstract}
In this work we recall the definition of matrix immanants, a generalization of the determinant and permanent of a matrix. We use them to generalize families of symmetric and antisymmetric orbit functions related to Weyl groups of the simple Lie algebras $A_n$.
The new functions and their properties are described, in particular we give their continuous orthogonality relations. Several examples are shown.
\end{abstract}

\maketitle

{\small
\noindent
$^1$ Department of Physics, Faculty of Nuclear Sciences and
Physical Engineering, Czech Technical University in Prague,
B\v{r}ehov\'a~7, 115 19 Prague 1, Czech Republic;\\
$^2$ Institute of Mathematics, University of Bialystok, Akademicka 2, PL-15-267 Bialystok, Poland

\medskip
\noindent
\textit{E-mail:} lenka.hakova@fjfi.cvut.cz, a.tereszkiewicz@uwb.edu.pl
}

\section{Introduction}

Immanants, i.e., matrix functions connected to irreducible characters of symmetric groups, were introduced and studied by Littlewood~\cite{Li}. They are widely used in algebraic combinatorics, e.g., immanants of certain matrices whose entries are symmetric functions define the Schur functions. The best known immanants are the determinant and the permanent of a matrix. The determinant is well known from its extended use in linear algebra.  Many results on the permanents can be found in the seminal book of Minc~\cite{Minc}. Other interesting applications are obtained in the connection with positive semi-definite Hermitian matrices. Recently they have been used in a mathematical description of three-channel optical networks~\cite{dG1,dG2}. In~\cite{dG3} the authors use immanants to analyze coincidence rates in connection with BosonSampling BosonSampling computation originally introduced in~\cite{AA}. Immanants also appear in statistical physics, see~\cite{MM}. See also recent experiments on permanents in~\cite{1,2,3,4}.

For every symmetric group $S_n$ there are several immanants defined, their numbers being equal to the number of corresponding conjugacy classes of the group. 
In this work we use immanants to generalize symmetric and antisymmetric orbit functions, families of multivariable complex functions related to Weyl groups of simple Lie algebras $A_n$. These functions can be written as the determinant and permanent of a particular matrix of order $n+1$, see~\cite{NPT}. By considering general immanants of the same matrix we obtain new families of functions. We study their properties, e.g.,  their symmetries with respect to elements of the Weyl group $A_n$, their products and continuous orthogonality.

Symmetric and antisymmetric orbit functions were studied in detail in~\cite{KP1,KP2}. Their orthogonality relations can be found in~\cite{MP06} and their discretization is completely described in~\cite{HrPa01}. Other families of orbit functions are defined in the connection of even Weyl subgroups and so-called sign homomorphisms~\cite{KP3,HHP,HHP2}.

The paper is organized as follows. Chapter 2 reviews some facts from the theory of characters of finite groups and the definition of immanants of a matrix. In Chapter 3 we present the connection between the symmetry group $S_{n+1}$ and the Weyl group of simple Lie algebra $A_n$, see, for example, Refs \cite{NPT,KP1}. Immanants functions of $W(A_n)$ are then defined in Chapter 4, including examples and properties for the particular case $n=2$. The main result lies in Chapter 5, where we generalize some properties of immanants functions and we prove Theorem~\ref{maintheorem} describing continuous orthogonality of immanant functions. We conclude with several remarks in Chapter 6.

\medskip

\section{Character tables and immanants}

\subsection{Irreducible characters of symmetric groups}

Immanants are a generalization of the determinant and permanent of a matrix defined using irreducible characters of symmetric groups. In this subsection we list some well known facts about irreducible characters which will be used in what follows. For a good review of the theory of immanants of finite groups see~\cite{JamLie}.

Let $G$ be a finite group. It can be written as a union of its conjugacy classes. Irreducible characters $\chi$ map each class to a complex number. The number of conjugacy classes equals the number of irreducible characters. The values of characters are then listed in so-called character tables, see Table~\ref{characters}.

Row orthogonality relations hold, i.e., for every irreducible characters $\chi_k,\chi_l$ we have
\begin{equation}\label{charOG}
\sum_{g\in G}\chi_k(g)\overline{\chi_l(g)}=|G|\delta_{kl},
\end{equation}
where $|G|$ denotes the order of the group $G$ and $\delta_{kl}$ is the Kronecker delta.

It also holds that for every element $g\in G$, $\chi(g^{-1})=\overline{\chi(g)}$. This implies that if every $g\in G$ is conjugated to its inverse then all the characters have real values. This is a case for example for all the symmetry groups.

Finally, we will need the following formula for the convolution of characters~\cite{Si}:
\begin{equation}\label{conv}
\sum_{g\in G}\chi_k(hg^{-1})\chi_l(g)=\delta_{kl}\frac{|G|}{d_k}\chi_k(h),
\end{equation}
where $d_k$ is the degree of the character $\chi_k$, i.e., $d_k=\chi_k(id)$.

\medskip

\subsection{Immanants}

Let $\mathcal{A}$ be a matrix of order $n$, $\mathcal{A}=(a_{ij})$. A product of entries of the form $a_{1i_{1}}a_{2i_{2}}\ldots a_{ni_{n}}$ corresponds to the element $\Pi$ of the symmetric group $S_n$ given by $\Pi(1,2,\ldots,n)=(i_1,i_2,\ldots,i_n)$. For each conjugacy class $[\rho]$ of $S_n$ we denote by $C_\rho$ the sum of all product of the matrix entries of the above form. The immanant corresponding to the irreducible character $\chi_k$ of $S_n$ is defined as
\begin{equation}\label{immanantdef}
 \Imm^{n,k}=\sum_{[\rho]} \chi_k(\rho) C_\rho,
\end{equation}
where $\chi_k(\rho)$ denotes the value of the character $\chi_k$ on the conjugacy class $[\rho].$

We list the immanants corresponding to matrices of order $2,3$ and $4$. Necessary character tables are listed in Table~\ref{characters} and can be found for example in~\cite{Li}.

{\small\begin{table}
\begin{tabular}{|c||c|c|}
\hline
 $[\rho]$ & (1) & (12) \\
\hline\hline
$\chi_1$ & 1 &1\\
\hline
$\chi_2$ & 1 & -1\\
\hline
\end{tabular}\hspace{20pt}
\begin{tabular}{|c||c|c|c|}
\hline
 $[\rho]$ & (1) & (12) & (123)\\
\hline\hline
$\chi_1$ & 1 &1&1\\
\hline
$\chi_2$ & 1 & -1&1\\
\hline
$\chi_3$ & 2 & 0&-1\\
\hline
\end{tabular}\hspace{20pt}
\begin{tabular}{|c||c|c|c|c|c|}
\hline
 $[\rho]$ & (1)& (12) & (12)(34) &(123)& (1234)\\
\hline\hline
$\chi_1$ & 1 &1&1&1&1\\
\hline
$\chi_2$ & 1 & -1&1&1&-1\\
\hline
$\chi_3$ & 2 & 0&2&-1&0\\
\hline
$\chi_4$ & 3 & 1&-1&0&-1\\
\hline
$\chi_5$ & 3 & -1&-1&0&1\\
\hline
\end{tabular}
\medskip
\caption{Character tables of $S_2,S_3$ and $S_4$. The standard cycle notation for the conjugacy class representant is used.}\label{characters}
\end{table}}

\begin{itemize}
 \item The group $S_2$ has two irreducible characters. The sums $C_1$ and $C_2$ are equal to
 $$\begin{aligned}
    C_{1}&=a_{11}a_{22},\\
    C_{2}&=a_{12}a_{21}.\\
   \end{aligned}$$
The immanant corresponding to the trivial representation gives in all cases the permanent of the matrix and the alternating representation gives its determinant. In the case of $n=2$ these are the only ones.
 $$\begin{aligned}
    \Imm^{2,1}&=C_{1}+C_{2}=a_{11}a_{22}+a_{12}a_{21},\\
    \Imm^{2,2}&=C_{1}-C_{2}=a_{11}a_{22}-a_{12}a_{21}.\\
   \end{aligned}$$

\item The character table of $S_3$ contains three characters. The sums $C_1, C_2$ and $C_3$ are then
 $$\begin{aligned}
    C_{1}&=a_{11}a_{22}a_{33},\\
    C_{2}&=a_{11}a_{23}a_{32}+a_{12}a_{21}a_{33}+a_{13}a_{22}a_{31},\\
    C_{3}&=a_{12}a_{23}a_{31}+a_{13}a_{21}a_{32}.\\
   \end{aligned}$$
Three irreducible characters induce three immanants. One is non-trivial in the sense that it is not equal to the determinant or permanent of the matrix.
 $$\begin{aligned}
    \Imm^{3,1}&=C_{1}+C_{2}+C_{3},\\
    \Imm^{3,2}&=C_{1}-C_{2}+C_{3},\\
    \Imm^{3,3}&=2C_{1}-C_{3}.\\
   \end{aligned}$$

\item There are five conjugacy classes of the group $S_4$ with the sums $C_1,\ldots,C_5$ equal to
 $$\begin{aligned}
    C_{1}&=a_{11}a_{22}a_{33}a_{44},\\
    C_{2}&=a_{11}a_{22}a_{34}a_{43}+a_{11}a_{23}a_{32}a_{44}+a_{11}a_{24}a_{33}a_{42}+a_{12}a_{21}a_{33}a_{44}\\ & +a_{13}a_{22}a_{31}a_{44}+a_{14}a_{22}a_{33}a_{41},\\
    C_{3}&=a_{12}a_{21}a_{34}a_{43}+a_{13}a_{24}a_{31}a_{42}+a_{14}a_{23}a_{32}a_{41},\\
    C_{4}&=a_{11}a_{23}a_{34}a_{42}+a_{11}a_{24}a_{32}a_{43}+a_{13}a_{22}a_{34}a_{41}+a_{14}a_{22}a_{31}a_{43}\\ & +a_{12}a_{24}a_{33}a_{41}+a_{14}a_{21}a_{33}a_{42}+a_{12}a_{23}a_{31}a_{44}+a_{13}a_{21}a_{32}a_{44},\\
    C_{5}&=a_{12}a_{23}a_{34}a_{41}+a_{12}a_{24}a_{31}a_{43}+a_{13}a_{24}a_{32}a_{41}+a_{13}a_{21}a_{34}a_{42}\\ & +a_{14}a_{23}a_{31}a_{42}+a_{14}a_{21}a_{32}a_{43}.\\
   \end{aligned}$$
There are five immanants, three of them non-trivial.
 $$\begin{aligned}
    \Imm^{4,1}&=C_{1}+C_{2}+C_{3}+C_{4}+C_{5},\\
    \Imm^{4,2}&=C_{1}-C_{2}+C_{3}+C_{4}-C_{5},\\
    \Imm^{4,3}&=2C_{1}+2C_{3}-C_{4},\\
    \Imm^{4,4}&=3C_{1}+C_{2}-C_{3}-C_{5},\\
    \Imm^{4,5}&=3C_{1}-C_{2}-C_{3}+C_{5}.\\
   \end{aligned}$$

\end{itemize}

\medskip

\section{Symmetric groups and Weyl groups}

\subsection{Symmetric group $S_{n+1}$}

Let the group $S_{n+1}$ act on the ordered number set $[l_1,l_2,\dots,l_n,l_{n+1}]$. We introduce an orthonormal basis in the real Euclidean space $\mathbb R^{n+1}$,
\begin{equation*}
{e_i}\in\mathbb R^{n+1}\,,\qquad
\l e_i , e_j\r=\delta_{ij}\,,\qquad    1\leq i,j\leq n+1\,,
\end{equation*}
and use the numbers $l_k$ as the coordinates of a point $\lambda$ in $e$-basis:
$$
\lambda=\sum_{k=1}^{n+1}l_ke_k\,,\qquad l_k\in\mathbb R\,.
$$

We consider the orbit of the action of $S_{n+1}$ on a point $\lambda\in\R^{n+1}$ and denote it by $O_{S_{n+1}}(\lambda)$, where  $\lambda$ is the unique point of this orbit such that
\begin{equation}\label{requirement2}
l_1\geq l_2\geq\cdots\geq l_n\geq l_{n+1}\,.
\end{equation}
If there are no coordinates $l_k$ in $\lambda$ equal, the orbit $O_{S_{n+1}}(\lambda)$ consists of
$(n+1)!$ points.

Subsequently we are interested only in points $\mu$ from the $n$-dimensional subspace
${\mathcal H}\subset\mathbb R^{n+1}$ defined by the requirement
\begin{align}\label{requirement}
\sum_{k=1}^{n+1}l_k=0\,.
\end{align}

\medskip

\subsection{Lie algebra~$A_n$}

Let us recall basic properties of the simple Lie algebra~$A_n$ of the compact Lie group $SU(n+1),$ where $1\leq n<\infty.$
The Coxeter-Dynkin diagram, Cartan matrix $  \mathfrak{C}$, and inverse Cartan matrix $  \mathfrak{C}^{-1}$ of~$A_n$  are the following ones:

\begin{gather*}
\parbox{.6\linewidth}
{\setlength{\unitlength}{1pt}
\def\kr{\circle{10}}
\thicklines
\begin{picture}(20,30)
\put(10,14){\kr}
\put(6,0){$\alpha_1$}
\put(15,14){\line(1,0){10}}
%
\put(30,14){\kr}
\put(26,0){$\alpha_2$}
\put(35,14){\line(1,0){10}}
%
\put(50,14){\kr}
\put(46,0){$\alpha_3$}
\put(55,14){\line(1,0){10}}
%
\put(70,13.5){$\ \,\ldots$}
%
\put(95,14){\line(1,0){10}}
%
\put(110,14){\kr}
\put(104,0){$\alpha_{n\!-\!1}$}
\put(115,14){\line(1,0){10}}
%
\put(130,14){\kr}
\put(126,0){$\alpha_n$}
\end{picture}
}
\hspace{-50 pt}
  \mathfrak{C}=\left(\begin{smallmatrix}
                          2&-1&0&0&\dots&0&0&0\\
                         -1&2&-1&0&\dots&0&0&0\\
                          0&-1&2&-1&\dots&0&0&0\\[-1ex]
\vdots&\vdots&\vdots&\vdots&\ddots&\vdots&\vdots&\vdots\\[0.5ex]
                          0&0&0&0&\dots&-1&2&-1\\
                          0&0&0&0&\dots&0&-1&2
                     \end{smallmatrix}\right),
\\[2ex]
  \mathfrak{C}^{-1}=\frac{1}{n+1}\left(\begin{smallmatrix}
1\cdot n\;    &1\cdot(n-1)\;&1\cdot(n-2)\;&\dots    &1\cdot 2\;   &1\cdot 1\\
1\cdot (n-1)\;&2\cdot(n-1)\;&2\cdot(n-2)\;&\dots    &2\cdot 2\;   &2\cdot 1\\
1\cdot (n-2)\;&2\cdot(n-2)\;&3\cdot(n-2)\;&\dots    &3\cdot 2\;   &3\cdot 1\\
\vdots        &\vdots           &\vdots       &\ddots&\vdots            &\vdots\\[0.5ex]
1\cdot 2\;    &2\cdot2\;    &3\cdot2\;       &\dots &(n-1)\cdot2\;&(n-1)\cdot 1\\
1\cdot 1\;    &2\cdot1\;    &3\cdot1\;        &\dots &(n-1)\cdot1\;&n\cdot 1\
                     \end{smallmatrix}\right).
\end{gather*}

The simple roots $\alpha_i$, $1\le i\le n$, of $A_n$ form a basis ($\a$-basis) of a real Euclidean space $\mathbb R^n$ and we choose them in  ${\mathcal H}$:
$$
\alpha_i = e_i - e_{i+1}, \quad i=1,\dots,n\,.
$$
Such a choice fixes the lengths and relative angles of the simple roots.
They are of the same length equal to~$\sqrt 2$ with relative angles between $\alpha_k$ and $\alpha_{k+1}$  $(1\leq k\leq n-1)$
equal to~$\frac{2\pi}{3}$, and $\tfrac\pi2$ for any other pair.

As usually,   we introduce  third basis, the $\w$-basis in ${\mathcal H}\cong\mathbb R^n\subset\mathbb R^{n+1}$,  as the $\mathbb Z$-dual one to the simple roots~$\alpha_i$:
$$
\l\alpha_i,\w_j\r=\delta_{ij}\,,\qquad 1\leq i\leq n.
$$
Vectors $\omega_i$ are called weights.
The set of all integer combinations of weights is denoted by $P$ and called the weight lattice,
$$P=\bigoplus_{i=1}^n\Z\omega_i.$$
The bases $\alpha$ and $\w$ are related by the Cartan matrix:
\begin{equation*}
\alpha=  \mathfrak{C}\w\,,\qquad \w=  \mathfrak{C}^{-1}\alpha.
\end{equation*}

Through the paper we assume $\lambda\in\mathcal H$ and  we fix the notation for its coordinates relative to the $\omega$-basis and euclidean basis:
\begin{equation}\label{notacja}
\lambda =\sum_{j=1}^{n+1}l_je_j=:(l_1,\ldots,l_{n+1})_e=
    \sum_{i=1}^n\lambda_i\w_i=:(\lambda_1,\ldots,\lambda_n)_\w\,, \quad\sum_{i=1}^{n+1}l_i=0.
     \end{equation}
The relations between coordinates in these basis are for example described in detail in~\cite{KP1}.

\medskip

\subsection{The Weyl group of $A_n$}

The Weyl group $W(A_n)$ acts on  ${\mathcal H}$ by permuting the coordinates in $e$-basis, i.e., it acts as the group $S_{n+1}$.
Indeed, let $r_i$, $1\le i\le n$ be generating elements  of $W(A_n)$,
reflections with respect to the hyperplanes perpendicular to $\alpha_i$ and passing through the origin. Then  the reflections $r_i$ correspond to adjacent transposition:
\begin{equation}\label{reflection}
\begin{aligned}
r_i \lambda&=\lambda-\langle \lambda,\alpha_i\rangle\alpha_i
=(l_1,l_2,\dots,l_{n+1})_e-(l_i-l_{i+1})(e_i-e_{i+1})\\
&=(l_1,\dots,l_{i-1},l_{i+1},l_i,l_{i+2},\dots,l_{n+1})_e.
\end{aligned}
\end{equation}
Because transpositions generate the full permutation group $S_{n+1}$, thus $W(A_n)$ is isomorphic to $S_{n+1}$, and the points of the orbit $O_{S_{n+1}}(\lambda)$ and $O_{W(A_n)}(\lambda)$ coincide.

The affine Weyl group of $A_n$ is the infinite group which can be described as a semidirect product of the translation by integer combinations of the simple roots $\alpha_i$ and the Weyl group, for more details see~\cite{HrPa01}. Its fundamental domain is a simplex $F$ defined as the convex hull of the vectors $\omega_i$.

\medskip

\section{New types of functions for $W(A_n)$}\label{newfunctions}

\subsection{Correspondence between orbit functions of $W(A_n)$ and immanants}

In this paper we are interested in the following matrix $\mathcal{A}:$
\begin{equation*}\label{matrix}
\mathcal{A}=
  \left(
\begin{array}{cccc}
e^{2\pi {\rm i}l_1x_1}&e^{2\pi {\rm i}l_1x_2}&\dots&e^{2\pi {\rm i}l_1x_{n+1}}\\
e^{2\pi {\rm i}l_2x_1}&e^{2\pi {\rm i}l_2x_2}&\dots&e^{2\pi {\rm i}l_2x_{n+1}}\\[-1ex]
\vdots&\vdots&\ddots&\vdots\\[1 ex]
e^{2\pi {\rm i}l_{n+1}x_1}&e^{2\pi {\rm i}l_{n+1}x_2}&\dots&e^{2\pi {\rm i}l_{n+1}x_{n+1}}
\end{array}
\right),
\end{equation*}
where the number $l_i$ fulfil the requirements given by~\eqref{requirement2} and~\eqref{requirement}.

This matrix is related to two families of Weyl group orbit functions of $A_n,$ see \cite{KP1,KP2,MP06,HrPa01}.
The symmetric orbit functions, also called (normalized) $C$-functions, for the group $W(A_n)$, are defined for every $x\in\R^n$
\begin{equation*}\label{cHAT}
\Phi_{\lambda}(x):= \sum_{w \in W(A_n)} e^{2\pi i \l w\lambda, x\r},
\end{equation*}
where $\lambda\in P$. With the notation introduced in~\eqref{notacja}, these functions can be written as the immanant of the matrix $\mathcal{A}$ corresponding to the trivial character.
$$
\Phi_{\lambda}(x)=\per\mathcal{A}=\Imm^{n+1,1}(\lambda,x).
$$
The $C$-orbit functions are invariant with respect to the Weyl group,
\begin{equation}\label{Cinv}
\Phi_{w\lambda}(x)=\Phi_{\lambda}(x)\end{equation}
for every $x\in\R^n$ and $w\in W (A_n)$. Therefore, it is enough to choose $\lambda\in P^+=\bigoplus_{i=1}^n\Z^{\ge 0}\omega_i$. We call such $\lambda$ a dominant point.

The second type are the antisymmetric or (normalized) $S$-orbit functions defined for every $x\in\R^n$ by
\begin{equation*}\label{sHAT}
\varphi_{\lambda}(x)= \sum_{w \in W(A_n)}( \det w)e^{2\pi {\rm i} \l w\lambda, x\r},
\end{equation*}
where  $\lambda\in P^+=\bigoplus_{i=1}^n\Z^{\ge 0}\omega_i$. In this case we have
$$
\varphi_{\lambda}(x)=\det\mathcal{A}=\Imm^{n+1,2}(\lambda,x),
$$
i.e., the function equals the immanant corresponding to the alternating character.
By the symmetries of $S-$orbit functions and the fact that they are identically zero on the boundary of the fundamental domain, we choose $\lambda\in P^{++}=\bigoplus_{i=1}^n\Z^{> 0}\omega_i$. In this case we call $\lambda$ a strictly dominant point.

These two families were studied in detail in several papers, e.g.,  \cite{KP1,KP2,MP06,HrPa01}. For the later use we give their continuous orthogonality relations~\cite{MP06}.

For every $\lambda,\mu\in P^+$ it holds that
\begin{equation}\label{COG}
\int_{F} \Phi_{\lambda}(x)\overline{\Phi_{\mu}(x)} dx=|F||W(A_n)||\stab_{W(A_n)}(\lambda)|\delta_{\lambda\mu},
\end{equation}
where $|F|$ denotes the volume of the fundamental domain $F$, $|W(A_n)|$ is the order of the Weyl group of $A_n$ and $\stab_{W(A_n)}(\lambda)$ is the stabilizer of the action of $W$ on $\lambda$.

Similarly, for every $\lambda,\mu\in P^{++}$ it holds that
\begin{equation}\label{SOG}
\int_{F} \varphi_{\lambda}(x)\overline{\varphi_\mu(x)} dx=|F||W(A_n)|\delta_{\lambda\mu}.
\end{equation}

In particular, if $\mu=0$, the $C$-function becomes a constant equal to $|W(A_n)|$ and the equation~\eqref{COG} gives
$$\int_{F} \Phi_{\lambda}(x)\overline{\Phi_{0}(x)} dx=|W(A_n)|\int_{F} \Phi_{\lambda}(x) dx=|F||W(A_n)|^2\delta_{\lambda 0}.$$
We can rewrite it in terms of immanants as
\begin{equation}\label{ImmOG}
\int_{F} \Imm^{n+1,1}(\lambda,x) dx=|F||W(A_n)|\delta_{\lambda 0}.
\end{equation}

In this work we are interested in the functions coming from the immanants $\Imm^{n+1,k}$, where $k\ge 3$. In this case we can use the correspondence between the reflections and adjacent transpositions given by~\eqref{reflection}. By $\left[\rho\right]$ we can denote not only the permutations belonging to one conjugacy class, but also a subset of the Weyl group generated by elements corresponding to these permutation. The elements can be found by decomposing each permutation into adjacent transpositions.
In terms of Weyl group the formula for immanants \eqref{immanantdef} becomes
\begin{equation*}\label{immanantWeyl}
\Imm^{n+1,k}(\lambda,x)=\sum_{[\rho]}\chi_{k}(\rho)\sum_{w\in[\rho]} e^{2\pi {\rm i} \l w\lambda,x\r}=\sum_{w\in W(A_n)}\chi_{k}(w) e^{2\pi {\rm i} \l w\lambda,x\r}.
\end{equation*}

\medskip

\subsection{Weyl group of $A_1$}

For this group according to tables of characters there are only two types of immanants:
\begin{eqnarray*}
\Imm^{2,1}(\lambda,x)&=&\Phi_\lambda(x),\\
\Imm^{2,2}(\lambda,x)&=&\varphi_\lambda(x)  .
\end{eqnarray*}
These two functions are up to a constant the cosine and sine functions or equivalently the Chebyshev polynomials.

\medskip

\subsection{Weyl group of $A_2$}

Group $W(A_2)$ is generated by two reflections, see \eqref{reflection}. It is a group of order $6$ and it can be decomposed in its conjugacy classes as follows:
$$
W(A_2)=\{id\}\cup\{r_1,r_2,r_1r_2r_1\}\cup\{r_1r_2,r_2r_1\}.
$$

For this group according to character tables there are three immanants:
\begin{eqnarray*}
\nonumber\Imm^{3,1}(\lambda,x)&=&\Phi_\lambda(x),\\
\Imm^{3,2}(\lambda,x)&=&\varphi_\lambda(x),\\
\nonumber\Imm^{3,3}(\lambda,x).&&
\end{eqnarray*}
We obtain the cosine and sine functions of two complex variables, see \cite{NPST}, and one new function, which is described in what follows.

Firstly, let us write the function $\Imm^{3,3}(\lambda,x)$ in the explicit form:
 \begin{equation*}
\Imm^{3,3}(\lambda,x)=2e^{2\pi {\rm i} (l_1x_1+l_2x_2+l_3x_3)}-e^{2\pi {\rm i} (l_1x_3+l_2x_1+l_3x_2)}-e^{2\pi {\rm i} (l_1x_3+l_2x_1+l_3x_2)}
 \end{equation*}
or, in terms of Weyl group
 \begin{equation*}
\Imm^{3,3}(\lambda,x)=2e^{2\pi {\rm i} \l\lambda,x\r}-e^{2\pi {\rm i}  \l r_2r_1\lambda,x\r}-e^{2\pi {\rm i} \l r_1r_2\lambda,x\r}.
 \end{equation*}
Unlike the $C-$ and $S-$ functions, this function is not invariant under the action of Weyl group. Nevertheless, the following relations hold
\begin{alignat*}{2}
\nonumber\Imm^{3,3}(w\lambda,x)&=\Imm^{3,3}(\lambda,w^{-1}x),\\
\Imm^{3,3}(\lambda,x)&=\Imm^{3,3}(w\lambda,w x),\\
\nonumber\Imm^{3,3}(\lambda,x)&=\Imm^{3,3}(x,\lambda).
\end{alignat*}

For the product with $C-$ and $S-$ functions we get by direct computation
\begin{equation*}
\Imm^{3,3}(\lambda,x)\Phi_{\mu}(x)=\sum_{w\in W(A_2)}\Imm^{3,3}(\lambda+w\mu,x),
\end{equation*}

\begin{equation*}
\Imm^{3,3}(\lambda,x)\varphi_{\mu}(x)=\sum_{w\in W(A_2)}\det w \;\Imm^{3,3}(\lambda+w\mu,x).
\end{equation*}
These relations come from general formula for product of immanants~\eqref{product}.

One of the most remarkable properties of orbit function is their continuous and discrete orthogonality. The main theorem describing the continuous orthogonality relations of immanant functions and its proof are contained in the chapter~\ref{orthogonality}. The following theorem is a special case for the functions $\Imm^{3,3}$.

\begin{theorem}
Let $F$ be the fundamental domain of the affine Weyl group of $A_2$ and we  denote the region $\widetilde{F}:=\bigcup\limits_{w\in W(A_2)}wF$. Then for every $0\neq\lambda,\mu\in P^{+}$ we have
\begin{equation*}
\int\limits_{\widetilde{F}} \Imm^{3,3}(\lambda,x)\overline{\Imm^{3,3}(\mu, x) } d x=|F||W(A_2)|^2\;\delta_{\lambda\mu}.
\end{equation*}
\end{theorem}

\medskip

\subsection{Examples of immanants of $W(A_2)$}

The function $\Imm^{3,3}(\lambda,x)$ is the zero function when $\lambda$ is the zero vector. For every other choice of $\lambda$ it is a function with non-zero real and imaginary part. In this paragraph we present several examples by plotting the real and the imaginary part of the function (see Figs.~\ref{fig1} and~\ref{fig2}).

\begin{figure}[ht!]
  \centering
 \includegraphics[width=0.3\textwidth]{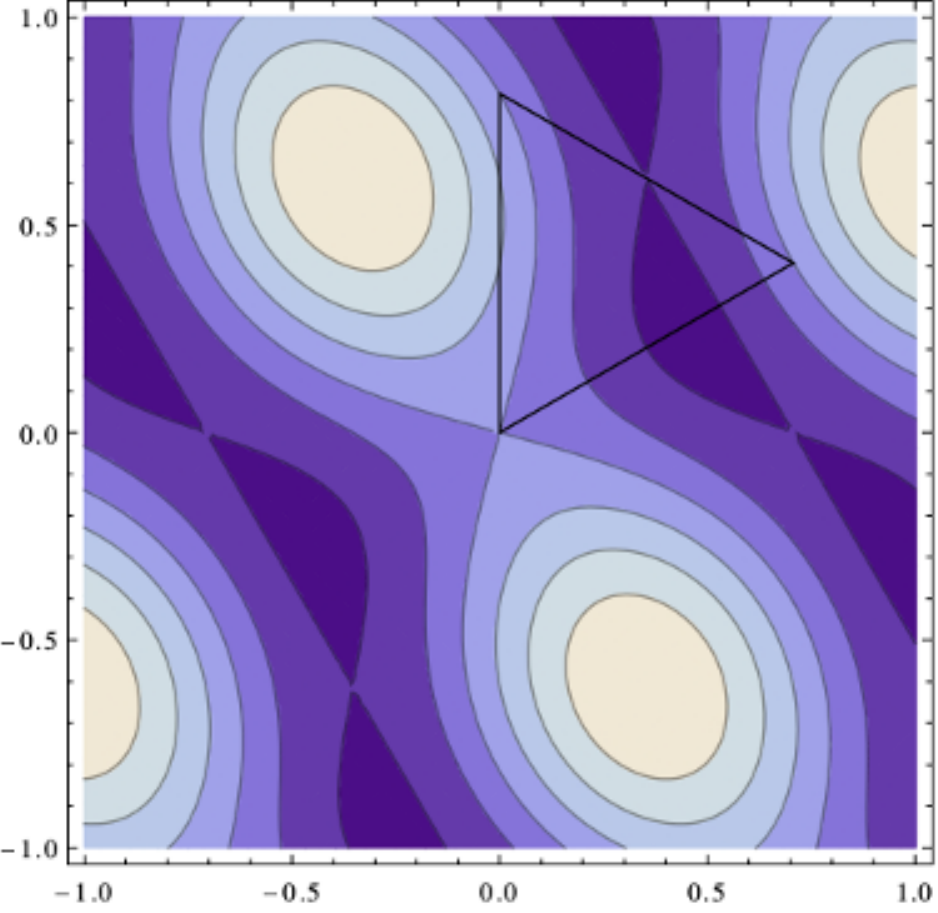}
  \includegraphics[width=0.3\textwidth]{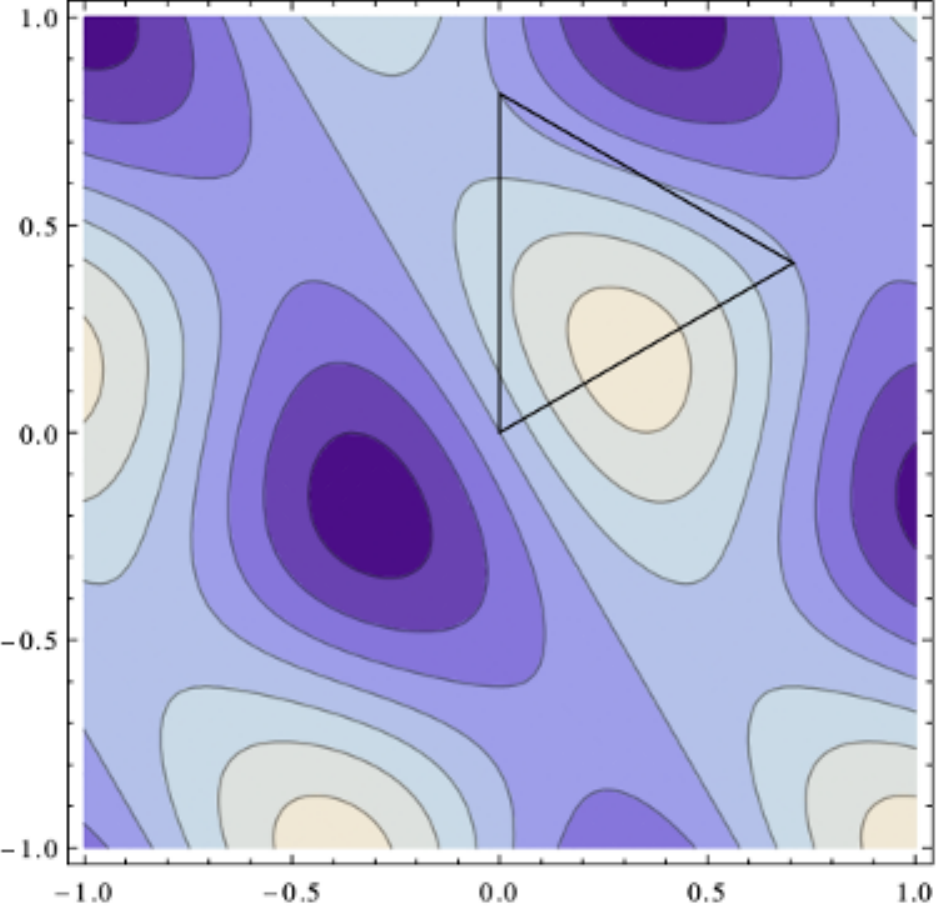}
  \caption{The contour plot of the real part (left) and the imaginary part (right) of the function $\Imm^{3,3}((1,0),x)$. The triangle denotes the fundamental domain $F$ of the affine Weyl group $W(A_2)$.}\label{fig1}
\end{figure}

\begin{figure}[ht!]
  \centering
  \includegraphics[width=0.3\textwidth]{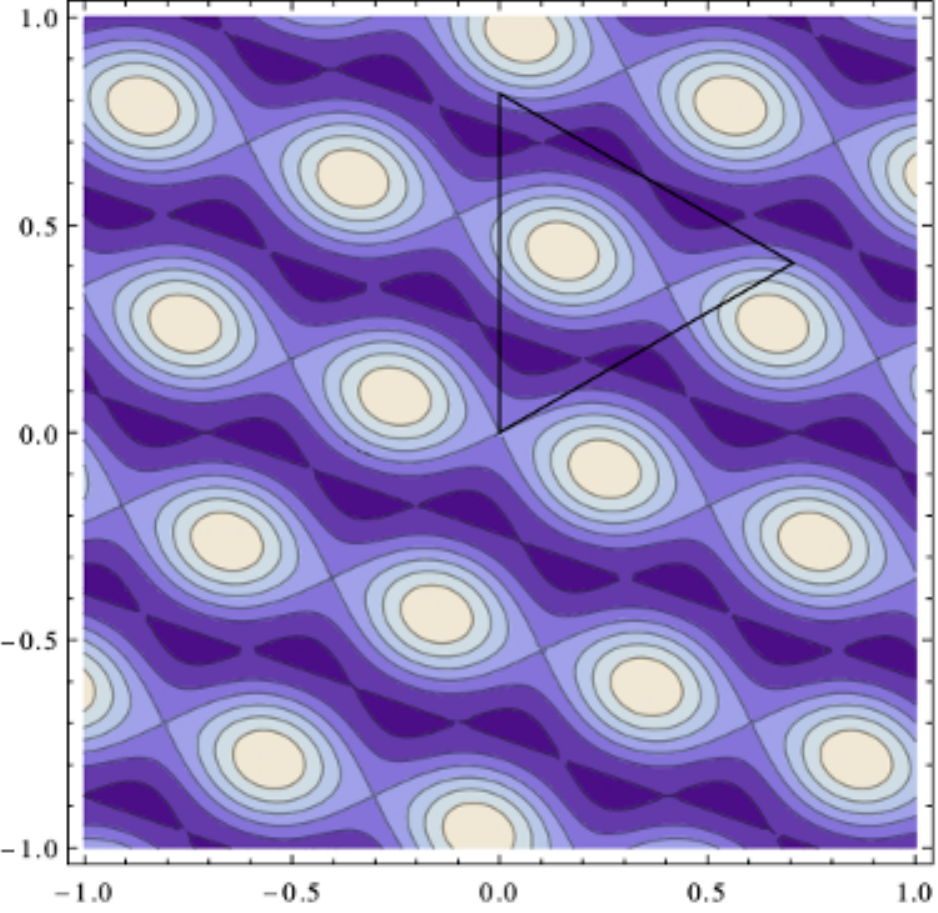}
 \includegraphics[width=0.3\textwidth]{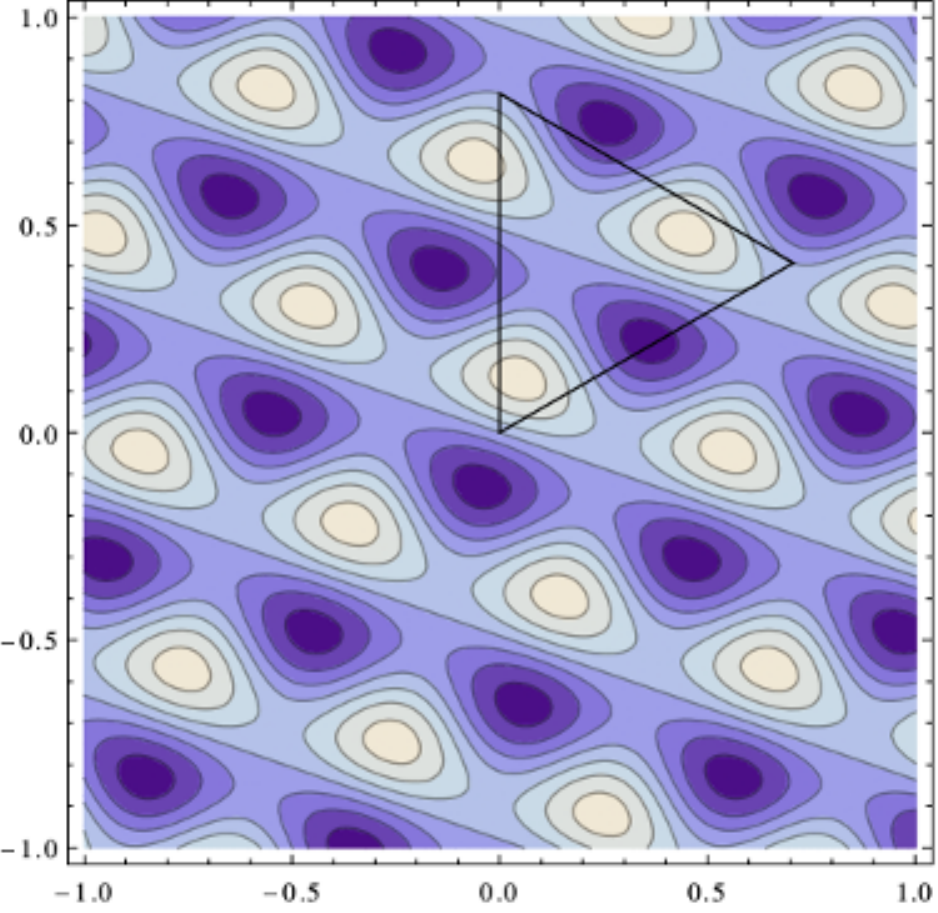}
  \caption{The contour plot of the real part (left) and the imaginary part (right) of the function $\Imm^{3,3}((1,2),x)$. The triangle denotes the fundamental domain $F$ of the affine Weyl group $W(A_2)$.}\label{fig2}
\end{figure}

\medskip

\subsection{Weyl group of $A_3$}

Group $W(A_3)$ is generated by three reflections, see \eqref{reflection}. It is a group of order $24$ and it can be decomposed into conjugacy classes as follows:
$$
\begin{aligned}
W(A_3)&=\{id\}\cup \{r_1,r_2,r_3,r_1r_2r_1,r_2r_3r_2,r_1r_2r_3r_2r_1\} \cup \{r_1r_3,r_2r_1r_3r_2,r_3r_2r_3r_1r_2r_3\}\\
&\cup \{r_1r_2,r_2r_1, r_2r_3,r_3r_2,r_1r_3r_2r_1, r_1r_2r_1r_3, r_2r_3r_2r_1,r_1r_2r_3r_2\}\\ &\cup \{r_1r_2r_3,r_2r_3r_1,r_3r_1r_2,r_3r_2r_1,r_3r_2r_3r_1r_2,r_1r_2r_1r_3r_2\}.
\end{aligned}$$

There are five immanants corresponding to the symmetric group $S_4$:
\begin{eqnarray*}
\Imm^{4,1}(\lambda,x)&=&\Phi_\lambda(x),\\
\Imm^{4,2}(\lambda,x)&=&\varphi_\lambda(x),\\
\Imm^{4,3}(\lambda,x),&&\\
\Imm^{4,4}(\lambda,x),&&\\
\Imm^{4,5}(\lambda,x).&&
\end{eqnarray*}
We obtain the cosine and sine functions of three complex variables, see \cite{NPST}, and three new functions. The implicit forms in the terms of Weyl groups are given below:
$$\begin{aligned}
\Imm^{4,3}(\lambda,x)&= 2e^{2\pi {\rm i} \l\lambda,x\r} +2e^{2\pi {\rm i} \l r_1r_3\lambda,x\r} +2e^{2\pi {\rm i} \l r_2r_1r_3r_2\lambda,x\r} +2e^{2\pi {\rm i} \l r_3r_2r_3r_1r_2r_3\lambda,x\r}\\
&-e^{2\pi {\rm i} \l r_1r_2\lambda,x\r} -e^{2\pi {\rm i} \l r_2r_1\lambda,x\r} -e^{2\pi {\rm i} \l r_2r_3\lambda,x\r}-e^{2\pi {\rm i} \l r_3r_2\lambda,x\r}\\
& -e^{2\pi {\rm i} \l r_1r_3r_2r_1\lambda,x\r} -e^{2\pi {\rm i} \l r_1r_2r_1r_3\lambda,x\r} -e^{2\pi {\rm i} \l r_2r_3r_2r_1\lambda,x\r} -e^{2\pi {\rm i} \l r_1r_2r_3r_2\lambda,x\r},\\
\Imm^{4,4}(\lambda,x)&= 3e^{2\pi {\rm i} \l\lambda,x\r} + e^{2\pi {\rm i} \l r_1\lambda,x\r} +e^{2\pi {\rm i} \l r_2\lambda,x\r} +e^{2\pi {\rm i} \l r_3\lambda,x\r} +e^{2\pi {\rm i} \l r_1r_2r_1\lambda,x\r} +e^{2\pi {\rm i} \l r_2r_3r_2\lambda,x\r}\\
& + e^{2\pi {\rm i} \l r_1r_2r_3r_2r_1\lambda,x\r}- e^{2\pi {\rm i} \l r_1r_3 \lambda,x\r}- e^{2\pi {\rm i} \l r_2r_1r_3r_2\lambda,x\r} -e^{2\pi {\rm i} \l r_3r_2r_3r_1r_2r_3\lambda,x\r} - e^{2\pi {\rm i} \l r_1r_2r_3\lambda,x\r}\\ &-e^{2\pi {\rm i} \l r_2r_3r_1\lambda,x\r} -e^{2\pi {\rm i} \l r_3r_1r_2\lambda,x\r}- e^{2\pi {\rm i} \l r_3r_2r_1\lambda,x\r} -e^{2\pi {\rm i} \l r_3r_2r_3r_1r_2\lambda,x\r} -e^{2\pi {\rm i} \l r_1r_2r_1r_3r_2\lambda,x\r} , \\
\Imm^{4,5}(\lambda,x)&= 3e^{2\pi {\rm i} \l\lambda,x\r} - e^{2\pi {\rm i} \l r_1\lambda,x\r} -e^{2\pi {\rm i} \l r_2\lambda,x\r} -e^{2\pi {\rm i} \l r_3\lambda,x\r} -e^{2\pi {\rm i} \l r_1r_2r_1\lambda,x\r} -e^{2\pi {\rm i} \l r_2r_3r_2\lambda,x\r}\\
& - e^{2\pi {\rm i} \l r_1r_2r_3r_2r_1\lambda,x\r}- e^{2\pi {\rm i} \l r_1r_3 \lambda,x\r}- e^{2\pi {\rm i} \l r_2r_1r_3r_2\lambda,x\r} -e^{2\pi {\rm i} \l r_3r_2r_3r_1r_2r_3\lambda,x\r} + e^{2\pi {\rm i} \l r_1r_2r_3\lambda,x\r}\\ &+e^{2\pi {\rm i} \l r_2r_3r_1\lambda,x\r} +e^{2\pi {\rm i} \l r_3r_1r_2\lambda,x\r}+ e^{2\pi {\rm i} \l r_3r_2r_1\lambda,x\r} +e^{2\pi {\rm i} \l r_3r_2r_3r_1r_2\lambda,x\r} +e^{2\pi {\rm i} \l r_1r_2r_1r_3r_2\lambda,x\r} . \\
\end{aligned}$$
 The symmetry properties and orthogonality relations follows from general results given in Chapter~\ref{CH5}.
\medskip

\section{Immanants of the Weyl group of $A_{n}$}\label{CH5}

\subsection{General properties of immanants of $W(A_{n})$}

Weyl group $W(A_{n})$ is generated by $n$ reflections and is isomorphic to the symmetry group $S_{n+1}$, with order $|W(A_{n})|=(n+1)!.$ Let $n_\chi$ denote the numbers of its conjugacy classes, i.e., also the number of corresponding irreducible characters and immanant functions. Two of them are the $C$ and $S$-orbit functions.

In what follows, we can suppose only $\lambda\neq 0$. Indeed, directly from the orthogonality of characters~\eqref{charOG} we see that for $\lambda=0$ and $k\in\{2,\ldots, n_\chi\}$, the immanant function $\Imm^{n+1,k}(\lambda,x)=0$ for every $x$; in the case of $k=1$ we have $\Imm^{n+1,1}(0,x)=|W(A_n)|$.

For every $k\in\{1,\ldots,n_\chi\}$ the following symmetries hold
\begin{equation}\label{sym}
\begin{array}{l}
\Imm^{n+1,k}(\omega\lambda,x)=\Imm^{n+1,k}(\lambda,\omega^{-1}x),\\
\Imm^{n+1,k}(\omega\lambda,\omega x)=\Imm^{n+1,k}(\lambda,x),\\
\Imm^{n+1,k}(\lambda,x)=\Imm^{n+1,k}(x,\lambda).
\end{array}
\end{equation}

We are interested in products of two different immanants. Direct computation give the following result. For every $\lambda,\mu\in P^+$ and $k\in\{1,\ldots,n_\chi\}$,
\begin{equation}\label{product}
\Imm^{n+1,k}(\lambda,x)\Imm^{n+1,l}(\mu,x) =\sum_{w,\tilde{w}\in W(A_n)}\chi_k(w)\chi_l(\tilde{w})  e^{2\pi {\rm i} \l\ w \lambda +\tilde{w}\mu,x\r}.
\end{equation}
In the case $l=1$ we get
\begin{equation*}
\Imm^{n+1,k}(\lambda,x)\Imm^{n+1,1}(\mu,x) =\sum_{\tilde{w}\in W(A_n)}\Imm^{n+1,k}(\lambda+\tilde{w}\mu,x).
\end{equation*}

For the proof of the orthogonality of immanants we need the following lemma.
\begin{lemma}\label{lemma}
Let $0\neq \lambda,\mu\in P^+$ and $k,l\in\{1,\ldots,n_\chi\}$. Then
\begin{equation*}
\begin{aligned}
\sum_{w\in W(A_n)}\Imm^{n+1,k}(w\lambda,x)\Imm^{n+1,l}(w\mu,x) &=\sum_{\tilde{w},\hat{w}\in W(A_n)}\chi_k(\tilde{w})\chi_l(\hat{w})\Imm^{n+1,1}(\lambda+\tilde{w}\hat{w}\mu,x),\\
\sum_{w\in W(A_n)}\Imm^{n+1,k}(w\lambda,x)\overline{\Imm^{n+1,l}}(w\mu,x) &=\sum_{\tilde{w},\hat{w}\in W(A_n)}\chi_k(\tilde{w})\chi_l(\hat{w})\Imm^{n+1,1}(\lambda-\tilde{w}\hat{w}\mu,x).\\
\end{aligned}
\end{equation*}
\end{lemma}

\begin{proof}
We have
$$\begin{aligned}
\sum_{w\in W(A_n)}\Imm^{n+1,k}(w\lambda,x)\Imm^{n+1,l}(w\mu,x)&=
\sum_{w,\tilde{w},\hat{w}\in W(A_n)} \chi_k(\tilde{w})\chi_l(\hat{w})e^{2\pi {\rm i} \l\left ( \tilde{w}w\lambda+\hat{w}w\mu\right ),x\r}\\
&=
\sum_{w,\tilde{w},\hat{w}\in W(A_n)} \chi_k(\tilde{w})\chi_l(\hat{w})e^{2\pi {\rm i} \l w\left ( w^{-1}\tilde{w}w\lambda+w^{-1}\hat{w}w\mu\right ),x\r}.\\
\end{aligned}$$

As the characters are class functions, we have $\chi\left (w^{-1}\tilde{w}w\right )=\chi\left (\tilde{w}\right )$ and $\chi\left (w^{-1}\hat{w}w\right )=\chi\left (\hat{w}\right )$, so we can write
$$\begin{aligned}
\sum_{w\in W(A_n)}\Imm^{n+1,k}(w\lambda,x)\Imm^{n+1,l}(w\mu,x)&=
\sum_{\tilde{w},\hat{w}\in W(A_n)} \chi_k(\tilde{w})\chi_l(\hat{w})\sum_{w\in W(A_n)}e^{2\pi {\rm i} \l w\left ( \tilde{w}\lambda+\hat{w}\mu\right ),x\r}\\
&=
\sum_{\tilde{w},\hat{w}\in W(A_n)} \chi_k(\tilde{w})\chi_l(\hat{w})\Imm^{n+1,1}\left (\tilde{w}\lambda+\hat{w}\mu,x\right )\\
&=
\sum_{\tilde{w},\hat{w}\in W(A_n)} \chi_k(\tilde{w})\chi_l(\hat{w})\Imm^{n+1,1}\left (\tilde{w}\left (\lambda+\tilde{w}^{-1}\hat{w}\mu\right ),x\right ).\\
\end{aligned}$$

Where the last equivalency is due to the invariance of $\Imm^{n+1,1}$, or the $C$-orbit functions, given in~\eqref{Cinv}. Using the fact, that $\chi_{k}(\tilde{w})=\chi_k(\tilde{w}^{-1})$ and simplifying the sum limits we finally get
$$\sum_{w\in W(A_n)}\Imm^{n+1,k}(w\lambda,x)\Imm^{n+1,l}(w\mu,x)=
\sum_{\tilde{w},\hat{w}\in W(A_n)} \chi_k(\tilde{w})\chi_l(\hat{w})\Imm^{n+1,1}\left (\lambda+\tilde{w}\hat{w}\mu,x\right ).$$

The proof of the second formula is analogous. We use the fact that the characters of symmetric groups are all real, therefore the complex conjugacy only changes the sign in the first argument of the immanant $\Imm^{n+1,1}$.
\end{proof}

\medskip

\subsection{Continuous orthogonality of immanants of $W(A_{n})$}\label{orthogonality}

The main result of the paper is the following theorem describing the continuous orthogonality of immanant functions of Weyl group of any $A_n$.

\begin{theorem}\label{maintheorem}
Let $W(A_n)$ be the Weyl group of $A_n$ and let $F$ be the fundamental domain of the corresponding affine Weyl group. We denote $\widetilde{F}=\bigcup_{w\in W(A_n)} w F$. Then for every $0\neq \lambda,\mu\in P^{+}$ and every $k,l\in\{1,\ldots,n_\chi\}$ the following relation holds.
\begin{equation*}
\int_{\widetilde{F}} \Imm^{n+1,k}(\lambda,x) \overline{\Imm^{n+1,l}(\mu,x)} dx=
|W(A_n)|^2|F|\delta_{\lambda\mu}\delta_{kl}\frac{1}{d_k} \sum_{w\in\stab{W(A_n)}(\lambda)} \chi_k(w)\;\;,
\end{equation*}
where $d_k=\chi_k(id)$.

In particular, for $\lambda,\mu\in P^{++}$ it holds that
\begin{equation*}
\int_{\widetilde{F}} \Imm^{n+1,k}(\lambda,x) \overline{\Imm^{n+1,l}(\mu,x)} dx=
|F||W(A_n)|^2\delta_{kl}\delta_{\lambda \mu}\;\;.
\end{equation*}
\end{theorem}

\begin{proof}

Using the symmetries of immanant functions~\eqref{sym} we can write
$$\begin{aligned}
\bigcup_{w\in W(A_n)}\int\limits_{F} \Imm^{n+1,k}(\lambda,x) \overline{\Imm^{n+1,l}(\mu,x)} dx
&=\sum_{w\in W(A_n)}\int\limits_{F} \Imm^{n+1,k}(w\lambda,x) \overline{\Imm^{n+1,l}(w\mu,x)} dx\\
&=\int\limits_{F} \sum_{w\in W(A_n)} \Imm^{n+1,k}(w\lambda,x) \overline{\Imm^{n+1,l}(w\mu,x)} dx.\\
\end{aligned}$$

From Lemma~\ref{lemma} we get
$$\begin{aligned}
\int\limits_{F} \sum_{w\in W(A_n)} \Imm^{n+1,k}(w\lambda,x) \overline{\Imm^{n+1,l}(w\mu,x)} dx&=\int\limits_{F}\sum_{\tilde{w},\hat{w}\in W(A_n)} \chi_k(\tilde{w})\chi_l(\hat{w})\Imm^{n+1,1}\left (\lambda-\tilde{w}\hat{w}\mu,x\right )dx\\
&= \sum_{w,\hat{w}\in W(A_n)} \chi_k(w\hat{w}^{-1})\chi_l(\hat{w})\int\limits_{F}\Imm^{n+1,1}\left (\lambda-w\mu,x\right )dx\;\;.
\end{aligned}$$
In the last equation we substituted $w=\tilde{w}\hat{w}$.

According to equation~\eqref{ImmOG}, we need to investigate when the expression $\lambda-w\mu$ equals zero. Since $\lambda,\mu\in P^+$, this can happen only for $\lambda=\mu$ and $w\in\stab_{W(A_n)}(\lambda)$. Using the relations~\eqref{ImmOG} and~\eqref{conv} we get
$$\begin{aligned}
\sum_{w,\hat{w}\in W(A_n)} \chi_k(w\hat{w}^{-1})\chi_l(\hat{w})\int\limits_{F}\Imm^{n+1,1}\left (\lambda-w\mu,x\right )dx &=|W(A_n)||F|\delta_{\lambda\mu} \sum_{w\in\stab_{W(A_n)}(\lambda)} \sum_{\hat{w}\in W(A_n)} \chi_k(w\hat{w}^{-1})\chi_l(\hat{w})\\
&=
|W(A_n)|^2|F|\delta_{\lambda\mu}\delta_{kl}\frac{1}{d_k} \sum_{w\in\stab_{W(A_n)}(\lambda)} \chi_k(w).
\end{aligned}$$

For $\lambda,\mu\in P^{++}$ the stabilizer is trivial and the sum in above expression equals $d_k$.

\end{proof}

\section{Concluding remarks}
\begin{itemize}
\item The formula giving the volume of $F$ for every Weyl group can be found in~\cite{HrPa01}.

\smallskip

\item The orthogonality relation given by Theorem~\ref{maintheorem} holds also for $C$ and $S$ orbit functions, as they are just immanant functions $\Imm^{n+1,1}$ and $\Imm^{n+1,2}$. Nevertheless, their orthogonality is fulfilled on smaller region, see~\eqref{COG} and~\eqref{SOG}.

\smallskip

\item A natural question to ask is whether one can write also discrete orthogonality relations for immanants functions similar to those for orbit functions.
\end{itemize}

\section{Acknowledgments}
The authors would like to express their very great appreciation to Hubert de Guise for the inspiration for this work and his valuable comments on the manuscript.

This publication was supported by the European social fund within the framework of realizing the project „Support of inter-sectoral mobility and quality enhancement of research teams at Czech Technical University in Prague“, CZ.1.07/2.3.00/30.0034.

\bibliographystyle{amsplain}

\end{document}